\documentclass[11pt]{article}
\usepackage[margin=1in]{geometry}
\usepackage{soul}
\usepackage{amsmath,amsthm,amssymb,amsfonts}

\usepackage[noend]{algorithm2e}
\DontPrintSemicolon
\RestyleAlgo{ruled}

\newtheorem{lemma}{Lemma}
\newtheorem{theorem}{Theorem}
\DeclareMathOperator*{\argmax}{arg\,max}
\DeclareMathOperator*{\argmin}{arg\,min}

\newcommand{\OPT}{\mathsf{OPT}}
\newcommand{\ALG}{\mathsf{ALG}}

\title{Quickly Determining Who Won an Election}
\author{
	Lisa Hellerstein\thanks{Department of Computer Science and Engineering, New York University Tandon School of Engineering, New York, United States. Email \texttt{lisa.hellerstein@nyu.edu}.  Partially supported by NSF Award IIS-1909335.}
	\and
	Naifeng Liu\thanks{Department of Computer Science, CUNY Graduate Center; Department of Economics, University of Mannheim. Email \texttt{nliu@mail.uni-mannheim.de}. Partially supported by the Teaching Fellowship at City University of New York.}
	\and
	Kevin Schewior\thanks{Department of Mathematics and Computer Science, University of Southern Denmark. Email \texttt{kevs@sdu.dk}. Partially supported by the Independent Research Fund Denmark, Natural Sciences, Grant DFF-0135-00018B.}
}
\date{}

\begin{document}

\maketitle

\begin{abstract}
This paper considers elections in which voters choose one candidate each, independently according to known probability distributions. A candidate receiving a strict majority (absolute or relative, depending on the version) wins. After the voters have made their choices, each vote can be inspected to determine which candidate received that vote.
The time (or cost) to inspect each of the votes is known in advance. The task is to (possibly adaptively) determine the order in which to inspect the votes, so as to minimize the expected time to determine which candidate has won the election. We design polynomial-time constant-factor approximation algorithms for both the absolute-majority and the relative-majority version. Both algorithms are based on a two-phase approach. In the first phase, the algorithms reduce the number of relevant candidates to $O(1)$, and in the second phase they utilize techniques from the literature on stochastic function evaluation to handle the remaining candidates. In the case of absolute majority, we show that the same can be achieved with only two rounds of adaptivity.                                                                                                                                                                                                                                                  
\end{abstract}

\section{Introduction}
We introduce the following problem. Consider an election with voters $1,\dots,n$ who each vote for a single candidate out of $1,\dots,d$. Each voter $i$ chooses each candidate $j$ with known probability $p_{i,j}\in(0,1)$, independently of all other voters. A candidate wins the election if they receive a strict majority of the votes which may be, depending on the model, an absolute or a relative majority. After the votes have been collected, it is possible to count the vote of any voter $i$, taking a known time $c_i\geq 0$ (which can also be thought of as cost). The present paper is concerned with the following question: In what order should votes be counted so as to minimize the expected time (or cost) until it is known who won the election? We allow adapting this order along the way. 

The above problem is a stochastic function-evaluation problem (for a survey, see~\cite{unluyurt2004sequential}): The votes can be regarded as random variables taking one of $d$ values, and the function maps these variables to the index of the winning candidate or, if no such candidate exists, to $0$.  When $d=2$,
results on such problems can be applied to our problem.  Note that for $d=2$, there is no difference between absolute and relative majority. If $n$ is odd, the problem for $d=2$ is somewhat simpler because function value $0$ is impossible. The function can then be seen as a so-called $k$-of-$n$ function, a Boolean function that is $1$ if and only if at least $k$ of the $n$ variables have value $1$. For these functions, a beautiful optimal strategy due to Salloum, Breuer, and (independently) Ben-Dov is known~\cite{salloum1984optimum,ben1981optimal}.
In what follows, we refer to this as the SBB strategy.  In our setting, it translates to the following: 
An interchange argument shows that, if it were known a-priori that candidate $j$ is the winner, it would be optimal to inspect votes in increasing order of $c_i/p_{i,j}$ ratios. Since the two strategies, for $j\in\{1,2\}$, have to count $\lfloor n/2\rfloor +1$ votes before they have verified the winner, there exists, by the pigeonhole principle, a vote that is contained in both corresponding prefixes. An unconditionally optimal strategy can safely inspect this vote and then recurse on the resulting subinstance. 
In the case that $d=2$ and $n$ is even, determining who is the winner of the election, or determining that there is no winner because of a tie, is equivalent to determining whether there are precisely $n/2$ votes for both candidates, and an optimal strategy can be derived in a similar way~\cite{acharya2011expected,gkenosis22stochastic}.

For our general setting, where $d$ can be greater than 2, we do not expect a similarly clean strategy.  For previously studied functions of Boolean variables that are 
slight extensions of $k$-of-$n$ functions, 
the stochastic evaluation problem is either known to be NP-hard or no polynomial-time algorithms\footnote{For an algorithm for our problem to run in polynomial time, we require that, in every situation, it can find the next vote to count in time polynomial in the input size. Note that the corresponding decision tree may have exponential size. For an in-depth discussion of such issues, we refer to the survey by \"Unl\"uyurt~\cite[Section 4.1]{unluyurt2004sequential}.} computing optimal strategies are known. A recent trend in stochastic function evaluation has been to design (polynomial-time) \emph{approximation} algorithms instead, i.e., algorithms whose expected cost can be bounded with respect to the expected cost of an optimal strategy, preferably by a constant factor~\cite{ghuge2022non,gkenosis22stochastic,grammel2022,plank2022simple,liu20226approximation,deshpande2016approximation}.

The problems considered in this paper have similarities to the stochastic function-evaluation problem for Boolean linear threshold functions.  
Constant-factor approximation algorithms have been developed for that problem, using two different approaches. 
The first approach reduces the problem to Stochastic Submodular Cover through the construction of a utility function and solves the problem with an algorithm called Adaptive Dual Greedy ~\cite{deshpande2016approximation}.
We show that this approach yields an $O(d)$-approximation for our problems (but it is conceivable that a more nuanced analysis could yield a better approximation factor).
The second approach generates a separate strategy for each of the possible values of the function, by solving multiple instances of a constrained maximization problem, and interweaves the resulting strategies together~\cite{ghuge2022non}.
Applying this approach and current analysis to our problems would involve interweaving at least $d$ separate strategies, leading to (at least) a linear dependence on $d$.

Instead, we develop a two-phase approach, where the first phase reduces the number of relevant candidates to $O(1)$ and the second phase handles the remaining candidates with some of the aforementioned techniques. As we show, this approach yields $O(1)$-approximation algorithms for both the absolute-majority and the relative-majority setting.

A strategy is considered particularly useful in practice if it is non-adaptive (e.g.,~\cite{goemans2006stochastic,grammel2022,ghuge2022non}), i.e., it considers votes in a pre-specified order until the function is evaluated. For the absolute-majority version, we give an $O(1)$-approximation algorithm with only two rounds of adaptivity. In counting rounds, we use the ``permutation'' definition:
In each round, votes are inspected in an order that is pre-specified for that round, until some stopping condition is reached. 
This definition is consistent with usage in a number of previous papers (e.g.,~\cite{agarwal2019stochastic,ghuge2021power}).

For future work, we leave it open whether there are non-adaptive $O(1)$-approximate strategies and, if so, whether they can be computed in polynomial time.

\subsection{Our Contribution}

We start with discussing the arguably simpler absolute-majority case. Here, a candidate wins if they receive more than $n/2$ votes. Let us first formalize the termination condition, i.e., the condition that we have enough information to determine the outcome of the election. The condition is fulfilled if it is either known that (i) some particular candidate has won or that (ii) it is no longer possible for any candidate to win. Clearly, (i) is fulfilled if, for some candidate $j$, the number of votes candidate $j$ has received is at least $\lfloor n/2 \rfloor + 1$ and (ii) is fulfilled if, for all candidates $j$, the number of votes received by candidates other than $j$ has reached at least $\lceil n/2\rceil$.

Such conditions can be turned into submodular functions which map any state $b$ (representing the values of the votes counted so far) to a non-negative integer value such that reaching a state whose value is equal to a certain goal value is equivalent to satisfying the condition. 
For instance, as noted before, we can conclude that
candidate $j$ is not the absolute-majority winner if and only if the number of votes for candidates other than $j$ has reached at least $\lceil n/2\rceil$. 
Denote the number of votes in state $b$ for candidates other than $j$, capped at $\lceil n/2\rceil$, by $g_j(b)$.
By previous techniques (cf.~\cite{deshpande2016approximation}), such conditions can be combined to obtain a submodular function $g$ with the property that, once $g$ has reached a certain value, the termination condition is reached. As mentioned before, Adaptive Dual Greedy can then be applied to the resulting submodular function, and the approximation bound on Adaptivity Dual Greedy can then be shown to yield an approximation factor of $O(d)$. For completeness, we detail this analysis in Appendix~\ref{appendix:ADG}.

A second approach is based on the previously mentioned approach for $d=2$. Note that, using the SBB strategy, it is possible to evaluate at minimum expected cost whether any fixed candidate $j$ wins (by what can be viewed as merging the other candidates into a single one). Clearly, the expected cost of this strategy does not exceed the expected cost of an optimal strategy, as an optimal strategy also has to evaluate whether candidate $j$ wins (and more than that). Since concatenating the corresponding strategies (where repeated votes are skipped) determines who is the winner, the resulting strategy has an approximation factor of at most $d$. Unfortunately, such an approximation factor seems inherent to an approach which consists of separate strategies for all $d$ candidates, even if a more sophisticated round-robin approach (e.g.,~\cite{allen2017evaluation,plank2022simple}) is used to merge them into a single strategy.

As a third approach, let us try something much more naive, namely inspecting votes in increasing order of cost. Not surprisingly, this alone does not yield a constant-factor approximation algorithm, as shown, e.g., by the following instance: Consider $d=2$ with $n$ odd, and the following three types of votes:  (i) $(n-1)/2$ votes $i$ with $p_{i,1}=1-\varepsilon$ and $c_i=\varepsilon$, (ii) $(n-1)/2$ votes $i$ with $p_{i,1}=\varepsilon$ and $c_i=1-\varepsilon$, and (iii) a special vote $i^\star$ with $p_{i^\star,1}=1-\varepsilon$ and $c_{i^\star}=1$. 
When one considers votes in increasing order of cost, 
for $\varepsilon\to 0$, with probability approaching 1 it will be necessary to inspect all $n$ votes (at a total cost approaching $(n+1)/2$).  However,   
one could instead first inspect the type (i) votes, and then the special vote, and with probability approaching 1 this would be sufficient (at a total cost approaching $1$).

Interestingly, a combination of the two latter approaches yields a constant-factor approximation algorithm. The crucial observation is that inspecting votes in increasing order of cost \emph{as long as more than two candidates can still win} does not incur more cost than a constant factor times the cost that an optimal strategy incurs. To get an intuition for why this is true, consider the case in which all costs are either $\varepsilon$ (for negligible $\varepsilon$) or $1$. After all cost-$\varepsilon$ votes have been inspected, resulting in state $b$, let $m_2$ be the second-smallest distance of a function $g_j(b)$  from its goal value $\lceil n/2 \rceil$. It is not difficult to see that an optimum strategy still has to inspect at least $m_2$ votes. On the other hand, assume that, after inspecting $2 m_2$ votes, there are still more than two functions $g_j(b)$ that have not reached their goal values yet. Noting that, for any pair of such goal functions, inspecting any vote adds $1$ to at least one of them readily yields a contradiction. Interestingly, using a more involved charging argument, we are able to extend such an argument to arbitrary cost.

Next, consider the resulting situation, in which there are only two candidates left that can still win, say, candidates $1$ and $2$. We can now use the SBB strategy to first evaluate whether candidate $1$ wins and then, again, to evaluate whether candidate $2$ wins. Similarly to the argument given before, it is easy to argue that the expected cost of this is no larger than twice the expected cost of an optimal strategy. The resulting approximation guarantee of the entire algorithm is therefore $4$. 
By replacing each instance of the SBB strategy with a $2$-approximation non-adaptive strategy during the process of evaluating if candidate $1$ wins and then if candidate $2$ wins, we obtain an algorithm with only three rounds of adaptivity at the cost of a larger approximation factor of $6$. Additionally, we give another algorithm that further reduces the rounds of adaptivity to $2$ with a slightly subtle analysis, but it increases the approximation factor to $10$.

We now discuss the relative-majority case. Here, a candidate wins if they receive more votes than any other candidate. Again, the termination condition is fulfilled if it is either known that (i) some particular candidate has won or that
(ii) it is no longer possible for any candidate to win. To formalize this, consider any state $b$ and any two candidates $j$ and $k$. Note that $j$ is known to have received more votes than candidate $k$ if the number of inspected votes for candidate $k$ plus the number of uninspected votes is not larger than the number of inspected votes for candidate $j$. Again, we can translate this to a goal-value function $g_{j,k}(b)$. Then, (i) is fulfilled if, for some candidate $j$, the $d-1$ functions $g_{j,k}(b)$ (for $k\in\{1,\dots,d\}\setminus\{j\}$) have reached their goal values.  For (ii), all votes have to be inspected. 

Interestingly, a similar approach works: Again, we inspect votes in increasing order of costs until only two candidates can win. We can bound the total cost of this phase by $4$ times the cost of an optimal strategy. In the second phase, assume candidates $1$ and $2$ can still win. We would like to evaluate whether $g_{1,2}(b)$ reaches its goal value. To do so, we use Adaptive Dual Greedy. We adapt the analysis of Deshpande et al.\ for Boolean linear threshold functions~\cite{deshpande2016approximation},
to handle the fact each vote may contribute $0$, $1$, or $2$ to $g_{1,2}(b)$. As a consequence, we show that the cost is at most $3$ times the expected cost of an optimal strategy. After this algorithm has been applied, it is either known that candidate $1$ is the winner, or candidate $2$ is the only candidate that can still possibly win. We use the SBB strategy to handle the latter case. The resulting total approximation factor is $8$.

\subsection{Further Related Work}

The stochastic function-evaluation problem we discussed is inspired by the seminal work of Charikar et al.~\cite{charikar2000query}, where each input value is originally unknown but can be revealed by paying an \emph{information price} $c_i$, and the goal is to design a query strategy to compute the value of function $f$ with minimum cost. However, in the work of Charikar et al., the distributions of input variables are unknown so the performance is measured using competitive analysis. In our \emph{stochastic} setting, we can compute the expected cost of query strategies since the distributions are known, which allows us to use the standard analysis in approximation algorithms.
Recently, Blanc et al.~\cite{blanc2021query} revisited the priced information setting of Charikar et al., and they focus on a similar model to the Stochastic Boolean Function Evaluation problem of Deshpande et al.~\cite{deshpande2016approximation}, but allow an $\epsilon$ error in the evaluation result. 

Most recent works related to stochastic discrete minimization problems concentrate on designing non-adaptive strategies~\cite{ghuge2022non,gkenosis2018stochastic,grammel2022} or strategies that require small rounds of adaptivity~\cite{agarwal2019stochastic,ghuge2021power}. Non-adaptive strategies have advantages in many practical applications. For example, they typically require less storage space and they are amenable to parallelization. In addition, the \emph{adaptivity gap} is studied to measure the performance difference between an optimal adaptive algorithm and an optimal non-adaptive algorithm. So if the adaptivity gap is small, non-adaptive strategies are usually more desired. However, there are situations where it is hard to design a non-adaptive strategy with good performance, e.g., when the adaptivity gap is large. Surprisingly, sometimes allowing a few rounds of adaptivity can greatly reduce the expected cost of an optimal non-adaptive strategy. Hence, it is interesting to design strategies that permit small rounds of adaptivity to keep a subtle balance between performance, storage space, and parallelizability.

Interestingly, economists study similar problems in deciding the winner of an election while the goal is different~\cite{bognar2015optimal,gershkov2009optimal}. For example, in one of the settings of finding a first-best mechanism, they seek to maximize the \emph{social welfare}, which is the expected utility gain of individuals if the value of $f$ matches the input of any coordinate, minus the expected cost spent in learning the input values that are needed to determine the value of $f$. In other words, this approach generates a complicated rule for determining the value of election results. Our approach can be considered as fixing the decision rule so the utility gain is constant, and now the goal is equivalent to minimizing the cost of learning the value of $f$. Hence, our algorithms can be potentially helpful in mechanism-design problems.

\subsection{Overview}

In Section~\ref{sec:prelim}, we give some needed definitions. In Sections~\ref{sec:absolute} and~\ref{sec:relative}, we give constant-factor approximation algorithms for the absolute-majority and the relative-majority cases, respectively. We conclude and state open problems in Section~\ref{sec:conclusion}.

\section{Preliminaries}\label{sec:prelim}
Let $n$ be the number of voters, and $d$ be the number of candidates.  We number the voters from $1$ to $n$, and the candidates from $1$ to $d$. We sometimes use $[k]$ to denote the set $\{1,\dots,k\}$.

For $X=(X_1,\ldots,X_n) \in \{1,\ldots,d\}^n$
we define $N_j(X)$ to be the number of entries $X_i$ of $X$ that are equal to $j$. 
The {\em absolute-majority function}
$f:\{1,\ldots,d\}^n \rightarrow \{0,1,\ldots,d\}$ is defined as follows:
\begin{align*}
    f(X) = \begin{cases}
        j & \text{ if } N_j(X)\geq \lfloor\frac{n}{2}\rfloor +1,\\
        0 & \text{ otherwise. }
    \end{cases}
\end{align*}
Here $f(X)=j$ means $j$ is an absolute-majority winner, meaning $j$ receives more than half the votes, and $f(X)=0$ means there is no absolute-majority winner.
The {\em relative-majority function} $f: \{1,\dots,d\}^n\rightarrow \{0,\dots,d\}$ is as follows:
\begin{align*}
    f(X)=\begin{cases}
        j & \text{ if } \forall k\in\{1,\dots,d\}\setminus\{j\},\ N_j(X)>N_k(X),\\
        0 & \text{ otherwise.}
    \end{cases}
\end{align*}
Similar to the previous function, the output of $f$ is either the winner of the election, or $0$ if there is no winner.  Here the winner must receive more votes than any other candidate.

The above two functions are symmetric functions, meaning that the results of the function only depend on the \emph{number} of $1$'s, \dots, $d$'s in the input. (This property is called anonymity in social-choice theory.)
We assume that the vote of voter $i$ is an independent random variable $X_i$, taking a value in the set $\{1,\ldots,d\}$.  For each candidate $j$, we use
$p_{i,j}$ to denote the probability that $X_i=j$.
We assume that each $p_{i,j} > 0$.

In our vote-inspection problems, 
we are given as input the $p_{i,j}$ values for each voter $i$ and candidate $j$.  
For each voter $i$, we are also given a value $c_i \geq 0$.  The only way to discover the value of $X_i$, the vote of voter $i$, is to ``inspect'' vote $i$, which incurs cost $c_i$.
(We will sometimes use the phrase ``test variable $X_i$'' instead of ``inspect vote $i$''.) 
 In other words, $c_i$ is the cost of inspecting the vote of voter $i$.
The problem is to determine the optimal (adaptive) order in which to sequentially inspect the votes,
so as to minimize the expected total cost of determining the winner of the election.  

An assignment $a$ to the $n$ variables $X_i$ is a vector $a\in \{1,\ldots,d\}^n$, where $a_i$ is the value assigned to $x_i$.  We use assignments to represent the values of the votes of the $n$ voters.
A partial assignment is a vector $b\in \{1,\ldots,d,*\}^n$, where $b_i=*$ indicates that the value of $X_i$ has not been determined.  In discussing vote inspection strategies, we use a partial assignment $b$ to represent the current state of knowledge regarding the values of the votes inspected so far, with $b_i$ equal to the value of voter $i$'s vote if it has been inspected already, and $b_i=*$ otherwise.  
In the same way as for $X$, we use $N_j(b)$ to denote the number of entries $i$ of $b$ such that $b_i=j$.  So $N_j(b)$ represents the number of inspected votes that are for candidate $j$.  

\section{Absolute Majority}
\label{sec:absolute}

We describe an adaptive 4-approximation algorithm for the absolute-majority version of our problem. 
The general idea is as follows.
Throughout, we keep a partial assignment $b$ that records the vote values known so far.
For each candidate $j$, let $g_j(b)$ denote the number of known votes for candidates other than $j$, capped at $\lceil n/2 \rceil$.
We can rule out $j$ as the absolute-majority winner of the election iff $g_j(b) = \lceil n/2 \rceil$.  
The algorithm has two phases.
In the first phase, we inspect the votes $i$ in increasing order of their costs $c_i$, until there are at most two candidates $j$
who could still get enough votes to win the election. Equivalently, Phase 1 ends when the third-smallest value of $g_j(b)$, for all $j$, is $\lceil n/2 \rceil$.
At this point, if we have enough information to determine the outcome of the election, we are done.
Otherwise, for each of the at most two remaining potential winners $j$, we need to determine whether the remaining uninspected votes include enough additional votes for candidate $j$ 
to make $j$ the absolute-majority winner.
(If neither do, then the election has no absolute-majority winner.)
We take one of these remaining potential winners, 
calculate the minimum number $k$ of votes they would need to win,
and use the SBB strategy for evaluating $k$-of-$n$ functions to determine whether this candidate is the winner.
If not, and if another candidate is still a potential winner, we again use the SBB strategy to check if that candidate is the winner.  If not, we know there is no winner.
Pseudocode for the algorithm is presented in Algorithm~\ref{alg:abs-maj-1} 
(note that we are more verbose than necessary regarding the output, for readability).

\begin{algorithm}[!t]
\caption{A $4$-apx.\ adaptive algorithm for evaluating the absolute-majority function
}\label{alg:abs-maj-1}
$b \gets \{*\}^n$\;
\While{\textup{the value of $f$ is not certain}}{
    \eIf{\textup{the $3$rd-smallest (ties broken arbitrarily)  $g_j(b)<\lceil n/2 \rceil$}}{
    test the next untested variable in increasing order of $c_i$\;
    update $b$\;
    }
    {
    $\alpha\gets \text{any}\ j\in\argmin_{j\in\{1,\dots,d\}}g_j(b)$\;
    $\beta\gets \text{any}\ j\in\argmin_{j\in\{1,\dots,d\}\setminus\{\alpha\}}g_j(b)$\;
    evaluate if $f=\alpha$ using the SBB strategy\;
    update $b$\;
    \eIf{$f= \alpha$}
    {\Return $f=\alpha$}
    {
        \eIf{$g_\beta(b)=\lceil n/2 \rceil$}
        {\Return $f=0$}
        {evaluate if $f=\beta$ using the SBB strategy\;
        \eIf{$f= \beta$}
        {\Return $f=\beta$}
        {\Return $f=0$}
        }
    }
    }
}\vspace{0.5mm}
\Return the value of $f$
\end{algorithm}

We now prove an approximation bound on the algorithm.

\begin{theorem}\label{thm:abs-maj-1}
    Algorithm~\ref{alg:abs-maj-1} is an adaptive $4$-approximation algorithm for evaluating the absolute-majority function.
\end{theorem}

\begin{proof}
    We analyze the two phases of Algorithm~\ref{alg:abs-maj-1}: Phase 1 ends when the third-smallest $g_j(b)$ equals $\lceil n/2 \rceil$; note that Algorithm~\ref{alg:abs-maj-1} cannot terminate before that happens. Phase 2 ends when the algorithm terminates.
    The key to our proof is to show that, in Phase 1, Algorithm~\ref{alg:abs-maj-1} spends at most $2$ times the cost of an optimal strategy.  

    We first consider the following situation.
    Suppose we are in the process of executing Algorithm~\ref{alg:abs-maj-1} and Phase 1 has not yet ended. Let $b'$ be the current value of partial assignment $b$, and without loss of generality, assume $N_1(b')\geq \cdots\geq N_d(b')$. Let $m_1(b),\dots,m_d(b)$ denote the distances of functions $g_1(b),\dots,g_d(b)$ to $\lceil n/2 \rceil$, i.e.,
    \begin{equation*}
        \forall j\in\{1,\dots,d\},\ \ m_j(b)=\lceil n/2 \rceil-\min\Bigg\{\sum_{\ell\in\{1,\dots,d\}\setminus\{j\}}N_\ell(b), \lceil n/2 \rceil\Bigg\}.
    \end{equation*}
    Since $N_1(b')\geq \cdots\geq N_d(b')$, we have
    $m_1(b')\geq \cdots\geq m_d(b')$. Since Phase 1 has not yet ended, $m_3(b')>0$.
    Let $S'$ denote the set of variables tested so far by Algorithm~\ref{alg:abs-maj-1}.

    We will show that the following two facts hold:
    \begin{itemize}
    \item Fact 1: During its execution, an optimal algorithm needs to test at least $m_2(b')$ of the variables not in $S'$.
    \item Fact 2: To finish Phase 1, Algorithm~\ref{alg:abs-maj-1} needs to test at most $m_2(b')+m_3(b')$ additional variables.
    \end{itemize}

    To prove Fact 1, note that if $f=j\neq 0$, then during its execution, an optimal algorithm needs to test at least $\lfloor n/2\rfloor +1-N_j(b')$ of the variables not in $S'$; if $f=0$, an optimal algorithm needs to test at least $m_1(b')$ variables not in $S'$.  
    Since
    \begin{equation*}
        m_2(b')=\lceil n/2 \rceil - N_1(b')-\sum_{j=3}^d N_j(b') \leq \lfloor n/2\rfloor +1-N_1(b')\leq \lfloor n/2\rfloor +1-N_j(b')
    \end{equation*}
    and $m_2(b')\leq m_1(b')$, we have shown that Fact 1 holds, regardless of the value of $f$.
    
    We now show Fact 2. 
    Suppose that after performing $m_2(b')+m_3(b')$ additional tests, Algorithm~\ref{alg:abs-maj-1} still hasn't terminated.
    Let $b''$ denote the partial assignment after the $m_2(b')+m_3(b')$ additional tests are performed.
    If $m_2(b'')\neq 0$, the number of $2$'s among those $m_2(b')+m_3(b')$ variables must be more than $m_3(b')$; if $m_3(b'')\neq 0$, the number of $3$'s among those variables must be more than $m_2(b')$. Those two conditions cannot occur simultaneously, and if one of $m_2(b''),m_3(b'')$ is not equal to $0$, the remaining $m_j(b'')$ for $j>3$ are definitely $0$.
    If both $m_2(b'')=m_3(b'')=0$, we can use this argument again for $m_4(b'')$ and $m_5(b'')$. If still not finished, we can use it for $m_6(b'')$ and $m_7(b'')$, etc. Hence, there are at most two $m_j(b'')$ that are not equal to 0 after testing the additional $m_2(b')+m_3(b')$ variables because only $m_1(b'')$ and one other $m_j(b'')$ could be greater than 0.  Thus Phase 1 has ended, proving Fact 2.   
    
    Having proved the above two facts, which relate to the status of Algorithm~\ref{alg:abs-maj-1} when it has not yet completed Phase 1,  
    we now analyze the total cost spent in Phase 1, relative to the total cost spent in an optimal algorithm.
    Let $\sigma$ denote the indices of variables such that
    \begin{equation*}
        c_{\sigma_1}\leq c_{\sigma_2}\leq\cdots\leq c_{\sigma_n}.
    \end{equation*}
    Let $k$ be the number of variables tested by Algorithm~\ref{alg:abs-maj-1} in Phase 1.  Thus the variables tested in Phase 1 are  $\{x_{\sigma_1},\dots,x_{\sigma_k}\}$. 
    
    Now let $b'$ be the partial assignment right after Algorithm~\ref{alg:abs-maj-1} tests $x_{\sigma_1},\dots,x_{\sigma_{k-1}}$. We define $\phi_1,\dots,\phi_d$ such that
    \begin{equation*}
        m_{\phi_1}(b')\geq m_{\phi_2}(b')\geq \cdots\geq m_{\phi_d}(b').
    \end{equation*}
    Since testing the next variable $x_{\sigma_k}$ will end Phase 1, $m_{\phi_3}(b')=1$, and because testing $\{x_{\sigma_1},...,x_{\sigma_{k-1}}\}$ does not provide enough information to determine the value of $f$, an optimal algorithm needs to test at least one variable from $\{x_{\sigma_k},\dots,x_{\sigma_n}\}$ during its execution. Hence, we can compare the most expensive variable tested by an optimal algorithm, which costs at least $c_{\sigma_k}$, with the last two variables tested by Algorithm~\ref{alg:abs-maj-1} in Phase 1, which cost $c_{\sigma_{k-1}}+c_{\sigma_k}$.

    Now, suppose $b'$ is the partial assignment right after testing $x_{\sigma_1},\dots,x_{\sigma_{k-3}}$, and consider the corresponding $\phi_1,\dots,\phi_d$.
    Because there are exactly $3$ variables tested in Phase 1 after $\{x_{\sigma_1},\dots,x_{\sigma_{k-3}}\}$, and by Fact 2,
    $m_{\phi_2}(b') +m_{\phi_3}(b')$ is an upper bound on the number of additional variables tested in Phase 1, we have that
    $3\leq m_{\phi_2}(b')+m_{\phi_3}(b')$. Since the $m_j$'s are integer-valued functions, we have $m_{\phi_2}(b')\geq 2$.  By Fact 1, an optimal algorithm needs to test at least $m_{\phi_2}(b')$ variables not in $\{x_{\sigma_{1}},\dots,x_{\sigma_{k-3}}\}$ during its execution, so it needs to test at least two variables from $\{x_{\sigma_{k-2}},\dots,x_{\sigma_n}\}$. Hence, we can compare the second-most expensive variable tested by an optimal algorithm, which costs at least $c_{\sigma_{k-2}}$, with the third-last and the fourth-last variables tested by Algorithm~\ref{alg:abs-maj-1} in Phase 1, which cost $c_{\sigma_{k-3}}+c_{\sigma_{k-2}}$.
    
    More generally, let $b'$ denote the partial assignment right after having tested variables $x_{\sigma_1},\dots,x_{\sigma_{k-\ell}}$ where $1\leq\ell\leq k-1$, and again consider the corresponding $\phi_1,\dots,\phi_d$. 
    Algorithm~\ref{alg:abs-maj-1} does $\ell$ tests in Phase $1$ after testing $x_{\sigma_1},\dots,x_{\sigma_{k-\ell}}$, 
    where $\ell \leq m_{\phi_2}(b') + m_{\phi_3}(b')\leq 2m_{\phi_2}(b')$ by Fact $1$. 
    Applying Fact 2 then yields that an optimal algorithm needs to test at least $m_{\phi_2}(b') \geq \lceil\ell/2\rceil $ variables from $\{x_{\sigma_{k-\ell+1}},\dots,x_{\sigma_n}\}$.
    Hence, the $\lceil\ell/2\rceil$th-most expensive variable tested by an optimal algorithm costs at least $c_{\sigma_{k-\ell+1}}$. On the other hand, the $\ell$th-last and the $(\ell+1)$st-last variable tested by Algorithm~\ref{alg:abs-maj-1} in Phase 1 cost $c_{\sigma_{k-\ell+1}}+c_{\sigma_{k-\ell}}$ in total.
    
    To conclude our analysis of Phase 1, the cost Algorithm~\ref{alg:abs-maj-1} spends in Phase 1 is upper-bounded by twice the total cost of the $\lceil k/2\rceil$ most expensive variables tested by an optimal algorithm, and hence it is also upper-bounded by twice the total cost of all variables tested by an optimal algorithm.

    After Phase 1, there are three remaining possibilities: $f=\alpha$, $f=\beta$ and $f=0$. Let $\OPT$ denote an optimal adaptive strategy for the problem. If $f=\alpha$, clearly, our algorithm queries the remaining variables optimally; if $f=\beta$, since we always need to prove $f\neq \alpha$, our algorithm spends at most $2E[\OPT]$ in Phase 2; if $f=0$, we need to show $f\neq\alpha$ and $f\neq \beta$, hence, our algorithm spends at most $2E[\OPT]$ in Phase 2 as well. Combining the cost we spent in Phase 1, Algorithm~\ref{alg:abs-maj-1} is a $4$-approximation algorithm for evaluating an absolute-majority function when $c_i> 0$.
\end{proof}

The SBB strategy is a highly adaptive strategy; it is easy to see that Algorithm~\ref{alg:abs-maj-1}
may require $\Omega(n)$ rounds of adaptivity. 
Motivated by this, we also sketch a modification of this algorithm such that the new algorithm needs only 3 rounds of adaptivity, but the approximation ratio is increased slightly from 4 to 6.
Specifically, we 
replace Phase 2 which contains the two runs of the SBB strategy. After Phase 1 is completed, if we are still uncertain about the result of the election, we can first evaluate if $\alpha$ is the absolute-majority winner by using the non-adaptive 2-approximation algorithm for evaluating a $k$-of-$n$ function, which was introduced in Gkenosis et al.~\cite{gkenosis2018stochastic}.
This algorithm performs a modified (cost-sensitive) round robin between two orderings, 
one in increasing order of $c_i/p_{i,\alpha}$, and the other in increasing order of $c_i/(1-p_{i,\alpha})$. The modified round robin is due to Allen et al.~\cite{allen2017evaluation} and we present the pseudocode in Appendix~\ref{appendix:10-approx} (as it is used as part of another algorithm).
It keeps track of the next vote to be inspected in each ordering, and the total cost incurred so far for each ordering. It computes, for each ordering, the sum of the cost incurred so far, plus the cost of the next test.  It performs the next test from the ordering for which this sum is smaller (breaking ties arbitrarily). We call this Phase 2.

Since the problem of evaluating $f$ requires determining whether or not $\alpha$ is the winner, the cost spent in Phase 2 is upper-bounded by twice the expected cost of evaluating the value of $f$. If we know $\alpha$ is not the absolute-majority winner after Phase 2, we start Phase 3 by evaluating if $\beta$ is the winner using the same algorithm in Gkenosis et al. This will again add at most twice the expected cost of evaluating $f$. Clearly, if both $\alpha$ and $\beta$ are not the winner, there is no absolute-majority winner in the election..

We thus have the following theorem.
\begin{theorem}
    There exists a $6$-approximation algorithm for evaluating the absolute-majority function with $3$ rounds of adaptivity.
\end{theorem}

We note that it is possible to reduce the rounds of adaptivity from 3 to 2 using a modified round-robin applied to 4 different non-adaptive orderings.  We give details in Appendix~\ref{appendix:10-approx}.

\section{Relative Majority}
\label{sec:relative}

We now turn our attention to the relative-majority version of our problem and present an $8$-approximation algorithm. Similar to our previous approach, we keep a partial assignment $b$ to track the values of inspected votes and also to indicate which vote values are not yet known.  
To show that a candidate $j$ has more votes than a candidate $k$ is equivalent to showing that
\begin{equation*}
    N_k(b)+N_*(b)<N_j(b).
\end{equation*}
The fact that $N_*(b)=n-\sum_{k\in\{1,\dots,d\}}N_k(b)$ motivates designing \begin{equation*}
    g_{j,k}(b)=\min\Bigg\{N_j(b)+\sum_{\ell\in\{1,\dots,d\}\setminus\{k\}}N_\ell(b),n+1\Bigg\}
\end{equation*}
where $g_{j,k}(b)=n+1$ iff candidate $j$ is guaranteed to have more votes than candidate $k$. Hence, to conclude that $j$ is a relative-majority winner, we must have $g_{j,k}(b)=n+1$ for all $k\in\{1,\dots,d\}\setminus\{j\}$. 
Our algorithm again has two phases. In Phase 1, we still inspect the votes in increasing order of their costs. We end Phase 1 when there are at most two candidates who can possibly win and let $\alpha,\beta$ denote the two candidates. 
Then, there are at most three possible results of the election: $\alpha$ wins, $\beta$ wins, or it is a tie. We will first determine if $\alpha$ wins or not. Recall that, if $\alpha$ wins, they must have more votes than $\beta$, which means $g_{\alpha,\beta}(b)$ needs to reach $n+1$ at some point. Since a vote for $\alpha$,$\beta$, or any candidate in $\{1,\dots,d\}\setminus\{\alpha,\beta\}$ contributes $2$, $0$, or $1$ to $g_{\alpha,\beta}$, respectively (until threshold $n+1$ is reached), we can use the Adaptive Dual Greedy algorithm introduced by Deshpande et al.~\cite{deshpande2016approximation} to evaluate whether 
$g_{\alpha,\beta}(b)$ reaches its threshold value $n+1$. If yes, clearly $\alpha$ is the winner.  If not, then we test the remaining variables in increasing order of $c_i/(1-p_{i,\alpha})$ until it is certain whether $\beta$ is the winner or there is a tie.
As shown below, if all remaining votes are for $\alpha$, then $\alpha$ and $\beta$ are tied, otherwise $\beta$ is the winner.
The pseudocode is presented in Algorithm~\ref{alg:RMF-S} (again, this is slightly more verbose than necessary for readability).

\begin{algorithm}[t]
\caption{An 8-apx.\ adaptive algorithm for evaluating the relative-majority function}\label{alg:RMF-S}

$b \gets \{*\}^n$\;
\While{\textup{the value of $f$ is not certain}}
{
    \eIf{\textup{there are more than $2$ candidates that can possibly win}}
    {
        test the next untested variable in increasing order of $c_i$\;
        update $b$\;
    }
    {
        $\alpha\gets \text{any}\ j\in\argmax_{j\in\{1,\dots,d\}}N_j(b)$\;
        $\beta\gets \text{any}\ j\in\argmax_{j\in\{1,\dots,d\}\setminus\{\alpha\}}N_j(b)$\;
        evaluate if $f=\alpha$ by using the Adaptive Dual Greedy algorithm (cf.~\cite{deshpande2016approximation})\;
        \eIf{$f=\alpha$}
        {
            \Return $f=\alpha$\;
        }
        {
            evaluate if $f=\beta$ by testing remaining variables in increasing order of $\frac{c_i}{1-p_{i,\alpha}}$\;
            \eIf{$f=\beta$}
            {
                \Return $f=\beta$
            }
            {
                \Return $f=0$
            }
        }
    }
}\vspace{.5mm}
\Return the value of $f$
\end{algorithm}

We now give an analysis of the algorithm.

\begin{lemma}\label{lemma:rmf-phase1}
    The cost that Algorithm~\ref{alg:RMF-S} spends in Phase 1 is at most $4$ times the cost of an optimal algorithm.
\end{lemma}
\begin{proof}
    For each $j\in\{1,\dots,d\}$ and $k\in\{1,\dots,d\}\setminus\{j\}$, we define $m_{j,k}(b)$ to be the distance from $g_{j,k}(b)$ to $n+1$, so we have
    \begin{equation*}
        m_{j,k}(b) = n+1-g_{j,k}(b),
    \end{equation*}
    and define $M_j(b)$ such that 
    \begin{equation*}
        M_j(b)=\max_{k\in\{1,\dots,d\}\setminus\{j\}} m_{j,k}(b),
    \end{equation*}
    which is the largest distance among the $d-1$ utility functions $g_{j,k}$ of candidate $j$ to $n+1$.
    
    Recall the definition of the relative-majority function $f$. We can see that, when $f=0$, which means it is a draw situation, clearly we need to test all variables. When there exists a relative-majority winner, let $b'$ be the partial assignment at any moment before Algorithm~\ref{alg:RMF-S} terminates, and let $j^*$ be an arbitrary element in the set
    $\argmin_{j\in\{1,\dots,d\}} M_j(b')$.
    Then, any optimal algorithm needs to test at least $\lceil M_{j^*}(b')/2\rceil$ variables from the untested variables since any single test can contribute at most $2$ to any $g_{j,k}$. Here we note that $M_{j^*}(b')>0$, otherwise, in $b'$, $j^*$ would have strictly more votes than all other candidates, and it would already be known that $f=j^*$.

    Now we define $\gamma,\delta\in\{1,\dots,d\}\setminus\{j^*\}$ such that
    \begin{equation*}
        \forall k\in\{1,\dots,d\}\setminus\{\gamma,\delta,j^*\}, \ \ M_{j^*}(b')=m_{j^*,\gamma}(b')\geq m_{j^*,\delta}(b')\geq m_{j^*,k}(b'),
    \end{equation*}
    which implies that
    \begin{equation*}
        \forall k\in\{1,\dots,d\}\setminus\{\gamma,\delta,j^*\},\ \ N_\gamma(b')\leq N_\delta(b')\leq N_k(b')
    \end{equation*}
    since $m_{j^*,k}(b')=n+1-g_{j^*,k}(b')$ for any $k\in\{1,\dots,d\}\setminus\{j^*\}$.

    Suppose Phase 1 has not ended, meaning that there are more than two candidates who can still win at this moment. Then, we claim that testing any $m_{j^*,\gamma}(b')+ m_{j^*,\delta}(b')$ additional variables will end Phase 1.

    We now prove this claim. Let $m_\gamma=m_{j^*,\gamma}(b')$, $m_\delta=m_{j^*,\delta}(b')$. Let $b''$ denote the partial assignment after testing any $m_\gamma+m_\delta$ additional variables. 
    Suppose there exists a $\mu\in\{1,\dots,d\}\setminus\{j^*\}$ such that $m_{j^*,\mu}(b'')> 0$. Then
    \begin{align*}
        &m_{j^*,\mu}(b'')> 0\\
        \Leftrightarrow\ \ &m_{j^*,\mu}(b') - (N_{j^*}(b'')-N_{j^*}(b'))-m_\gamma-m_\delta+(N_\mu(b'')-N_\mu(b')) >0\\
        \Leftrightarrow\ \ &N_\mu(b'')-N_\mu(b')>m_\gamma+m_\delta+N_{j^*}(b'')-N_{j^*}(b')-m_{j^*,\mu}(b'')\\
        \Leftrightarrow\ \ &N_\mu(b'')-N_\mu(b')>(n+1-g_{j^*,\gamma}(b'))+(n+1-g_{j^*,\delta}(b'))\\
        &\hspace{6.2cm}+N_{j^*}(b'')-N_{j^*}(b')-(n+1-g_{j^*,\mu}(b'))\\
        \Leftrightarrow\ \ &N_\mu(b'')-N_\mu(b')>(n+1)-g_{j^*,\gamma}(b')-g_{j^*,\delta}(b')+g_{j^*,\mu}(b')+N_{j^*}(b'')-N_{j^*}(b').
    \end{align*}
    Similarly, if there exists a $\nu\in\{1,\dots,d\}\setminus\{j^*,\mu\}$ such that $m_{j^*,\nu}(b'')>0$, we have
    \begin{equation*}
        N_\nu(b'')-N_\nu(b')>(n+1)-g_{j^*,\gamma}(b')-g_{j^*,\delta}(b')+g_{j^*,\nu}(b')+N_{j^*}(b'')-N_{j^*}(b').
    \end{equation*}
    Then, we can see that if both $m_{j^*,\mu}(b'')> 0$ and $m_{j^*,\nu}(b'')>0$ hold, we have
    \begin{align*}
        &N_\mu(b'')-N_\mu(b')+N_\nu(b'')-N_\nu(b')\\
        >\;&2(n+1-g_{j^*,\gamma}(b')-g_{j^*,\delta}(b'))+g_{j^*,\mu}(b')+g_{j^*,\nu}(b')+2(N_{j^*}(b'')-N_{j^*}(b'))\\
        \geq\;&2(n+1)-g_{j^*,\gamma}(b')-g_{j^*,\delta}(b')+2(N_{j^*}(b'')-N_{j^*}(b'))\\
        \geq\;&m_\gamma+m_\delta,
    \end{align*}
     where the second inequality comes from $g_{j^*,\mu}(b')+g_{j^*,\nu}(b')\geq g_{j^*,\gamma}(b')+g_{j^*,\delta}(b')$ and the third inequality comes from $N_{j^*}(b'')-N_{j^*}(b')\geq 0$.  We have thus shown that 
     $N_{\mu}(b'')-N_{\mu}(b')+N_{\nu}(b'')-N_{\nu}(b') > m_\gamma+m_\delta$, which is a contradiction because 
     it implies that more than $m_\gamma+m_\delta$ additional votes were inspected. 
     Therefore, there is at most one $k$ such that $m_{j^*,k}(b'')>0$.  This proves the claim that inspecting any $m_{\gamma} + m_{\delta}$ additional votes will end Phase 1.
     
     By definition, $m_{j^*,k}(b'')=0$ is equivalent to $g_{j^*,k}(b'')=n+1$, which means $j^*$ has strictly more votes than $k$. This eliminates the possibility that $k$ is the winner. Since there exists at most one $k^*$ that satisfies $m_{j^*,k^*}\neq 0$, there are at most two candidates, $j^*$ and $k^*$, who can possibly win.

    Now we can use a similar argument as in the proof of Theorem~\ref{thm:abs-maj-1} to prove the rest of the lemma. 
    Let $\sigma_1,\dots,\sigma_n$ denote the indices of variables such that
    $c_{\sigma_1}\leq c_{\sigma_2}\leq\cdots\leq c_{\sigma_n}$. In an arbitrary realization, we assume at the time Phase 1 ends, Algorithm~\ref{alg:RMF-S} has tested $x_{\sigma_1},\dots,x_{\sigma_k}$.

    Let $b'$ denote the partial assignment when we have obtained the values of $x_{\sigma_1},\dots,x_{\sigma_{k-\ell}}$ where $\ell \geq 1$. From our previous claim, we know that
    \begin{equation*}
        \ell \leq m_{j^*,\gamma}(b')+m_{j^*,\delta}(b')\leq 2\cdot m_{j^*,\gamma}(b').
    \end{equation*}
    On the other hand, we know that an optimal algorithm needs to test at least $\lceil m_{j^*,\gamma}(b')/2\rceil$ variables outside of $\{x_{\sigma_1},\dots,x_{\sigma_{k-\ell}}\}$, which implies that an optimal algorithm needs to test at least $\lceil \ell /4\rceil$ variables from $\{x_{\sigma_{k-\ell+1}},\dots,x_n\}$. Hence, for example, by setting $\ell=1$, we can compare the most expensive variable tested by an optimal algorithm, which costs at least $c_{\sigma_k}$, with the total costs of $\{x_{\sigma_{k-3}},x_{\sigma_{k-2}},x_{\sigma_{k-1}},x_{\sigma_{k}}\}$; by setting $\ell=5$, we can compare the second-most expensive variable tested by an optimal algorithm, which costs at least $c_{\sigma_{k-4}}$, with the total costs of $\{x_{\sigma_{k-7}},x_{\sigma_{k-6}},x_{\sigma_{k-5}},x_{\sigma_{k-4}}\}$. Generalizing these observations, we can see that the $\lceil\ell/4\rceil$th-most expensive variable tested by any optimal algorithm costs at least $c_{\sigma_{k-\ell+1}}$. Hence, the total cost that Algorithm~\ref{alg:RMF-S} spends in Phase 1,
    $c_{\sigma_1}+\cdots+c_{\sigma_k}$, is at most 4 times the cost of an optimal algorithm.
\end{proof}

\begin{theorem}
    Algorithm~\ref{alg:RMF-S} is an adaptive 8-approximation algorithm for evaluating the relative-majority function.
\end{theorem}
\begin{proof}
    Let $\OPT$ denote an optimal algorithm.
    By using Lemma~\ref{lemma:rmf-phase1}, we know that Algorithm~\ref{alg:RMF-S} spent at most $4E[\OPT]$ in Phase 1. Hence, we just need to solve the induced problem after Phase 1.

    Let $\alpha,\beta$ denote the two candidates who can possibly win. We know that evaluating function $f$ will also evaluate if $f=\alpha$. Hence, the expected cost of an optimal algorithm for evaluating if $\alpha$ wins is upper-bounded by the expected cost of an optimal algorithm for evaluating the result of the election.
    Since there are only two candidates $\alpha,\beta$ remaining, evaluating if $f=\alpha$ is equivalent to evaluating if $\alpha$ has more votes than $\beta$. This problem can be understood as evaluating a variant of a linear threshold formula that was studied by Deshpande et al.~\cite{deshpande2016approximation}.
    In particular, let $b'$ denote the partial assignment right after we figured $\alpha,\beta$, we need to determine whether or not the following inequality, with variables $y_i$, holds:
    \begin{equation*}
        N_\alpha(b')+\Bigg(\sum_{k\in\{1,\dots,d\}\setminus\{\beta\}}N_k(b')\Bigg)+\sum_{i:b_i'=*}y_i \geq n+1.
    \end{equation*}
     Here $y_i=2$ if we find $x_i=\alpha$; $y_i=1$ if $x_i\in\{1,\dots,d\}\setminus\{\alpha,\beta\}$; and $y_i=0$ if $x_i=\beta$.  This is a linear threshold evaluation problem, involving ternary rather than binary variables. 
     The Adaptive Dual Greedy algorithm of Deshpande et al.\ can be used to solve this ternary linear threshold function evaluation problem, and determine whether $f = \alpha$.  A slight modification of the analysis used by Deshpande et al.\ shows that it spends at most $3E[\OPT]$. We defer details of the use of Adaptive Dual Greedy and its approximation bound to Appendix~\ref{appendix:ADG}.

    Suppose $f\neq \alpha$, what remains is to decide if $f=0$ or $f=\beta$. Recall that at the point we know $f\neq \alpha$, let $b'$ be the partial assignment, we have
    \begin{equation*}
        N_\alpha(b')+N_*(b')\leq N_\beta(b').
    \end{equation*}
    This means that, when the equality does not hold, we know $f=\beta$, since even if all the remaining variables have value $\alpha$, $\beta$ will have strictly more votes. Therefore, the induced problem can be considered as evaluating a conjunction (i.e., checking whether $x_i=\alpha$ for all untested $x_i$'s), which can be solved optimally by using the SBB strategy that tests the remaining variables in increasing order of $c_i/(1-p_{i,\alpha})$.

    Hence, Algorithm~\ref{alg:RMF-S} spends at most $4E[\OPT]+3E[\OPT]+E[\OPT]=8E[\OPT]$.
\end{proof}

\section{Conclusion}
\label{sec:conclusion}

In this paper, we have introduced the problem of quickly determining the winner of an election. We have also given the first constant-factor approximation algorithms for this problem in both the absolute-majority and the relative-majority case. While the approximation guarantees are one-digit numbers, we have no reason to believe that they match some type of approximation hardness. In fact, even NP-hardness remains an open problem. Assuming our problem is NP-hard, proving this is the case might require new techniques.  
Some stochastic evaluation problems are known to be NP-hard, for example, evaluation of Boolean linear threshold functions (cf.~\cite{unluyurt2004sequential}) and evaluation of $s$-$t$ connectivity in general graphs~\cite{fu2017determining}.  However,
it is open whether a number of Stochastic Score function evaluation problems, related to our problem, are NP-hard~\cite{acharya2011expected,gkenosis22stochastic,grammel2022}.  The question of whether the stochastic evaluation problem for Boolean read-once formulas is NP-hard has been open since the 1970's (cf.~\cite{unluyurt2004sequential}).

A structural question left open by our work is whether the adaptivity gap of our problem is constant. The strongest lower bound on the adaptivity gap that we are aware of comes from the lower bound of $1.5$ for $k$-of-$n$ functions~\cite{plank2022simple}. Extending the corresponding construction, one may try to consider an instance in which there are some votes of negligible cost that determine a forerunner, who can then be proved to be a winner cheaply by an adaptive strategy (by selecting the corresponding votes) but not by a non-adaptive strategy (because it does not know which votes to inspect). The extension is, however, not straightforward. One is restricted to certain distributions of who is the forerunner, and even votes for another candidate make progress on goal functions associated with other candidates.

Many natural extensions and variants of our model are possible, e.g., to other voting systems. A relatively small extension would concern larger target numbers of votes than $n/2$, possibly even candidate-specific ones. Such scenarios can be approached by adding an appropriate number of $0$-cost votes for each candidate (i.e., with probability close enough to $1$) and then using our algorithm for the absolute-majority version. Similarly, one could consider weighted voting systems (with modern applications in liquid democracy or voting in shareholder meetings). Ranked-voting systems could also be investigated.

\bibliographystyle{alpha} 
\bibliography{bibliography}

\newcommand{\etalchar}[1]{$^{#1}$}
\begin{thebibliography}{BBMtV15}

\bibitem[AAK19]{agarwal2019stochastic}
Arpit Agarwal, Sepehr Assadi, and Sanjeev Khanna.
\newblock Stochastic submodular cover with limited adaptivity.
\newblock In {\em ACM-SIAM Symposium on Discrete Algorithms (SODA)}, pages
  323--342, 2019.

\bibitem[AHK{\"U}17]{allen2017evaluation}
Sarah~R Allen, Lisa Hellerstein, Devorah Kletenik, and Tongu{\c{c}}
  {\"U}nl{\"u}yurt.
\newblock Evaluation of monotone dnf formulas.
\newblock {\em Algorithmica}, 77(3):661--685, 2017.

\bibitem[AJO11]{acharya2011expected}
Jayadev Acharya, Ashkan Jafarpour, and Alon Orlitsky.
\newblock Expected query complexity of symmetric {Boolean} functions.
\newblock In {\em Allerton Conference on Communication, Control, and Computing
  (Allerton)}, pages 26--29, 2011.

\bibitem[BBMtV15]{bognar2015optimal}
Katalin Bognar, Tilman B{\"o}rgers, and Moritz Meyer-ter Vehn.
\newblock An optimal voting procedure when voting is costly.
\newblock {\em Journal of Economic Theory}, 159:1056--1073, 2015.

\bibitem[BD81]{ben1981optimal}
Yosi Ben-Dov.
\newblock Optimal testing procedures for special structures of coherent
  systems.
\newblock {\em Management Science}, 27(12):1410--1420, 1981.

\bibitem[BLT21]{blanc2021query}
Guy Blanc, Jane Lange, and Li-Yang Tan.
\newblock Query strategies for priced information, revisited.
\newblock In {\em ACM-SIAM Symposium on Discrete Algorithms (SODA)}, pages
  1638--1650, 2021.

\bibitem[CFG{\etalchar{+}}00]{charikar2000query}
Moses Charikar, Ronald Fagin, Venkatesan Guruswami, Jon Kleinberg, Prabhakar
  Raghavan, and Amit Sahai.
\newblock Query strategies for priced information.
\newblock In {\em ACM Symposium on Theory of Computing (STOC)}, pages 582--591,
  2000.

\bibitem[DHK16]{deshpande2016approximation}
Amol Deshpande, Lisa Hellerstein, and Devorah Kletenik.
\newblock Approximation algorithms for stochastic submodular set cover with
  applications to {Boolean} function evaluation and min-knapsack.
\newblock {\em ACM Transactions on Algorithms (TALG)}, 12(3):1--28, 2016.

\bibitem[FFX{\etalchar{+}}17]{fu2017determining}
Luoyi Fu, Xinzhe Fu, Zhiying Xu, Qianyang Peng, Xinbing Wang, and Songwu Lu.
\newblock Determining source–destination connectivity in uncertain networks:
  Modeling and solutions.
\newblock {\em IEEE/ACM Transactions on Networking}, 25(6):3237--3252, 2017.

\bibitem[GGHK18]{gkenosis2018stochastic}
Dimitrios Gkenosis, Nathaniel Grammel, Lisa Hellerstein, and Devorah Kletenik.
\newblock The stochastic score classification problem.
\newblock In {\em European Symposium on Algorithms (ESA)}, pages 36:1--36:14,
  2018.

\bibitem[GGHK22]{gkenosis22stochastic}
Dimitrios Gkenosis, Nathaniel Grammel, Lisa Hellerstein, and Devorah Kletenik.
\newblock The stochastic {Boolean} function evaluation problem for symmetric
  {Boolean} functions.
\newblock {\em Discrete Applied Mathematics}, 309:269--277, 2022.

\bibitem[GGN21]{ghuge2021power}
Rohan Ghuge, Anupam Gupta, and Viswanath Nagarajan.
\newblock The power of adaptivity for stochastic submodular cover.
\newblock In {\em International Conference on Machine Learning (ICML)}, pages
  3702--3712, 2021.

\bibitem[GGN22]{ghuge2022non}
Rohan Ghuge, Anupam Gupta, and Viswanath Nagarajan.
\newblock Non-adaptive stochastic score classification and explainable
  halfspace evaluation.
\newblock In {\em International Conference on Integer Programming and
  Combinatorial Optimization (IPCO)}, pages 277--290, 2022.

\bibitem[GHKL22]{grammel2022}
Nathaniel Grammel, Lisa Hellerstein, Devorah Kletenik, and Naifeng Liu.
\newblock Algorithms for the unit-cost stochastic score classification problem.
\newblock {\em Algorithmica}, 84(10):3054--3074, 2022.

\bibitem[GS09]{gershkov2009optimal}
Alex Gershkov and Bal{\'a}zs Szentes.
\newblock Optimal voting schemes with costly information acquisition.
\newblock {\em Journal of Economic Theory}, 144(1):36--68, 2009.

\bibitem[GV06]{goemans2006stochastic}
Michel Goemans and Jan Vondr{\'a}k.
\newblock Stochastic covering and adaptivity.
\newblock In {\em Latin American Theoretical Informatics Symposium (LATIN)},
  pages 532--543, 2006.

\bibitem[Liu22]{liu20226approximation}
Naifeng Liu.
\newblock Two 6-approximation algorithms for the stochastic score
  classification problem.
\newblock {\em CoRR}, abs/2212.02370, 2022.

\bibitem[PS22]{plank2022simple}
Benedikt~M. Plank and Kevin Schewior.
\newblock Simple algorithms for stochastic score classification with small
  approximation ratios.
\newblock {\em CoRR}, abs/2211.14082, 2022.

\bibitem[SB84]{salloum1984optimum}
Salam Salloum and Melvin~A Breuer.
\newblock An optimum testing algorithm for some symmetric coherent systems.
\newblock {\em Journal of Mathematical Analysis and Applications},
  101(1):170--194, 1984.

\bibitem[{\"U}nl04]{unluyurt2004sequential}
Tongu{\c{c}} {\"U}nl{\"u}yurt.
\newblock Sequential testing of complex systems: a review.
\newblock {\em Discrete Applied Mathematics}, 142(1-3):189--205, 2004.

\end{thebibliography}

\newpage
\appendix

\section{Adaptive Dual Greedy} \label{appendix:ADG}

The Adaptive Dual Greedy algorithm of Deshpande et al.\ is an approximation algorithm for the
Stochastic Submodular Cover problem~\cite{deshpande2016approximation}.
After presenting a few definitions,
we begin by defining the Stochastic Submodular Cover problem and presenting the approximation bound achieved by the Adaptive Dual Greedy (ADG) algorithm.  We then show how to use ADG to solve the absolute-majority problem, and how to apply the approximation bound for ADG 
to prove an approximation bound for the absolute-majority problem.

Deshpande et al.\ 
used ADG to get a 3-approximation algorithm for the Stochastic Boolean Function Evaluation problem for Boolean linear threshold functions. 
The analysis proving the approximation factor of 3 does not immediately apply to our use of ADG in Algorithm~\ref{alg:RMF-S}, because there we need to evaluate a linear threshold function with inputs in $\{0,1,2\}$, rather than with Boolean inputs.
At the end of this section, we verify that 
the approximation factor of 3 also holds for our ternary linear threshold evaluation problem.

\subsection{Definitions}
For two partial assignments $b,b'\in ([d] \cup \{*\})^n$, we say $b'$ is an extension of $b$, denoted by $b'\succeq b$, if
for all $i\in [n] \text{ where } b_i\neq *$, it holds that $b_i'=b_i$.

We say that a partial assignment $b \in ([d] \cup \{*\})^n$ is a $z$-certificate for a function $f$ defined on $[d]^n$ if for all assignments $a \in [d]^n$, 
$a\succeq b$ implies that $f(a)=z$.
We write $f(b)=z$ to denote that partial assignment $b$ is a $z$-certificate. Furthermore, we say $b$ is a certificate of function $f$, if $b$ is a $z$-certificate of function $f$ for any output value $z$ that function $f$ can have. If $a$ is a fixed assignment in $[d]^n$, we say that a subset $S$ of $\{x_1,\dots,x_n\}$ is a certificate for the function $f$ (for assignment $a$) if the partial assignment $b$ with $b_i=a_i$ for all $x_i$ in $S$, and $b_i=*$ for all $x_i$ not in $S$, is a certificate for $f$.

\subsection{Stochastic Submodular Cover}

The inputs to the Stochastic Submodular Cover problem include a utility function $g:([d] \cup \{*\})^n \rightarrow \mathbb{Z}^{\geq 0}$, and a positive integer $Q$ that we call the goal value.  The function $g$ is given by an oracle, and has the following properties:  (1) it is monotone and submodular, according to the definitions given below,   (2) $g(*,\ldots,*)=0$, and (3) for all $x \in [d]^n$, $g(x)=Q$.
The other inputs to the problem are costs $c_i \geq 0$ for $i \in [n]$, and probabilities $p_{i,j}$ for $i \in [n]$ and $j \in [d]$, such that $p_{i,j} \in (0,1)$.
As in the absolute-majority problem, we assume that the value of each input $x_i$ to $g$ is a 
random variable, whose value is determined by performing a test with cost $c_i$.  The probability that $x_i=j$ is $p_{i,j}$.
We use partial assignment $b \in ([d] \cup *)^ n$ to represent the outcomes of the tests performed so far.   
The problem is to find an adaptive strategy for achieving $g(b) = Q$ with minimum expected cost.  

We say that
utility function $g$ is {\em monotone} if $g(b') \geq g(b'')$ whenever $b' \succeq b''$.
We say it is {\em submodular} if for any two partial assignments $b' \succeq b''$, $i \in [n]$ such that $b'_i=b''_i = *$, and $j \in [d]$, we have
$g(b'_{i \leftarrow j}) - g(b') \leq g(b''_{i \leftarrow j}) - g(b'')$.  Here the subscript means that the partial assignment is modified by changing the $i$th entry of the partial assignment to $j$.

We refer the reader to the paper of Deshpande et al.\ for pseudocode for ADG, and a description of how it works~\cite{deshpande2016approximation}.  Here, we discuss the approximation bound proved by Deshpande et al.\ for this algorithm.
They proved an upper bound on the approximation factor achieved by Adaptive Dual Greedy, with respect to the optimal (minimum expected cost) adaptive strategy for achieving goal value $Q$ for utility function $g$. The expression for the bound is somewhat complicated.  
To describe it, we first need some notation.

For $x \in [d]^n$, let $C(x)$ be the sequence of bits $x_i$ that are tested when running ADG on $x$.  Define $\mathbb{P}(x)$ to be the set of all proper prefixes of this sequence. (We treat each prefix as the set of elements it contains, since the ordering within the prefix is not relevant to us.)

For full assignment $x \in [d]^n$ and $S \subseteq [n]$, let $g(x,S)$ denote the value of $g$ on the partial assignment $b \in ([d] \cup \{*\})^n$ such that $b_i=*$ for all $i \not \in S$, and $b_i=x_i$ for all $i \in S$.

Deshpande et al.\ proved that the expected cost of the ADG strategy, for achieving goal utility $Q$ for utility function $g$, 
is within the following factor of the optimal adaptive strategy:

\begin{equation}
\label{ADGbound}
{\mathcal A}:=\max_{x \in [d]^n} \max_{S \in \mathbb{P}(x)} \frac{\sum_{i \in C(x)-S} [g(x,S \cup \{i\}) - g(x,S)]}{Q-g(x,S)}.
\end{equation}

Intuitively, this bound can be interpreted as follows.
Consider running ADG on some fixed $x$ in $[d]^n$.
Consider a partial assignment $b$ that represents the outcomes of the tests performed so far, at some point in the running of ADG on $x$.
Let $S$ denote the set of items tested so far, and let $S'$ denote the additional items that will be tested by ADG before $g$ achieves goal value $Q$ on assignment $x$.

At this point, $g(b)$ utility has been achieved, and to reach the goal, an additional $Q-g(b)$ units of utility must be achieved (this is the distance to goal).  
We know that we can achieve that additional $Q-g(b)$ utility for realization $x$ by adding the results of the tests in $S'$ to the results of the tests in $S$.
That is, $g(x,S \cup S')-g(x,S) = Q - g(x,S)$.
In words, the cumulative increase in utility, if we add the results of all the tests in $S'$ to the results we already obtained for the tests in $S$, is $Q-g(x,S)$.
Now consider beginning with the results of the tests in $S$ and adding only the result of a single additional test $i$.  This will cause an increase in utility of $g(x,S \cup \{i\})-g(x,S)$.
If $g$ were a modular function, then we would have that
$\sum_{i \in S'} (g(x,S \cup \{i\}) - g(x,S)) = g(x,S \cup S') - g(x,S)$.  But because $g$ is submodular, the quantity on the left hand side of this equation can also be strictly larger (but not smaller) than the quantity on the right.
The ADG bound is the maximum, over all $x$ and all associated subsets $S$
and $S'$ (produced by starting to run ADG on $x$ and stopping at any point) of the ratio between the quantity on the left hand side, and the quantity on the right.

\subsection{Using ADG to solve Absolute Majority}

We use the same approach used by Deshpande et al.\ in solving the problem of evaluating a Boolean linear threshold function~\cite{deshpande2016approximation}.
We reduce the problem of evaluating an absolute-majority function to Stochastic Submodular Cover through the construction of a ``goal function'' $g$ for the absolute-majority problem.
This is a utility function $g:([d] \cup *)^n \rightarrow \mathbb{Z}^{\geq 0}$ obeying the  properties listed above, and 
such that for $b \in ([d] \cup \{*\})^n$, $g(b) = Q$ iff $b$ is a certificate for the absolute-majority function.

There are multiple ways that one could construct a goal function $g$ for the absolute-majority function.  Here we present one such construction.
For each candidate $j$, let $g_j$ be the utility function such that on partial assignment $b$, $g_j(b) = \min \{\lfloor n/2 \rfloor+1, N_j(b)\}$.
We define the goal value of this function to be $Q_j:=\lfloor n/2 \rfloor +1$.
Partial assignment $b$ is a certificate that $j$ won the absolute-majority election iff $g_j(b)$ is equal to its goal value, $Q_j=\lfloor n/2 \rfloor +1$.

For each candidate $j$, we also define a goal function $g_{j0}$, with goal value $Q_{j0}= \lceil n/2 \rceil$.  The value of $g_{j0}(b)$ is 
$\min\{\lceil n/2 \rceil, \sum_{k \in [d]\backslash \{j\}} N_k(b)\}$,
the total number of entries of $b$ that are equal to values in $[d]$ other than $j$.  So, the value of $g_{j0}(b)$ is the number
of votes seen so far that are for candidates other than $j$, capped at $\lceil n/2\rceil$.  Partial assignment $b$ is a certificate that candidate $j$ is not the absolute-majority winner iff $g_{j0}(b)=Q_{j0} = \lceil n/2 \rceil$.

We now use standard AND and OR constructions to produce our final goal function $g$ (cf. ~\cite{deshpande2016approximation}).

We use the OR construction to combine the $g_j$ functions into a single utility function $\hat{g}$.
It reaches its goal value when at least one of the $g_j$ reaches its goal value.
The goal value of $\hat{g}$ is $\hat{Q}=\prod_{j\in[d]} Q_j$, and its value on a partial assignment $b$ is as follows:

$$\hat{g}(b) = \hat{Q} - \prod_{j\in[d]}(Q_j-g_j(b)).$$

We combine the $g_{j0}$'s into a single utility function $g_0$, using the standard AND construction.  The value of $g_0$ on a partial assignment $b$ is
 
 $$g_0(b) = \sum_{j \in [d]} g_{j0}(b).$$  

The goal value of $g_0$ is $Q_0=\sum_{j \in [d]} Q_{j0} = 
d \lceil n/2 \rceil$.  Utility function $g_0$
 reaches its goal value on a partial assignment $b$ iff
all the $g_{j0}$'s have reached their goal values on partial assignment $b$.

We then apply the OR construction to $\hat{g}$ and $g_0$ to create our final
utility function $g$.  The goal value of $g$ is $Q=\hat{Q}Q_0$ and its value on partial assignment $b$ is as follows:
$$g(b) = Q-(\hat{Q}-\hat{g}(b))(Q_0 - g_0(b)).$$

Utility function $g$ reaches its goal value
iff either $\hat{g}(b)=\hat{Q}$ (meaning $b$ is a certificate that some candidate has
an absolute majority) or $g_0(b) = Q_0$ (meaning $b$ is a certificate that no candidate has an absolute majority).  
The problem of finding an adaptive strategy of minimum expected cost for evaluating the absolute-majority function is equivalent to the problem of finding an adaptive strategy of minimum expected cost for achieving goal value $Q$ for utility function $g$.
Hence running Adaptive Dual Greedy, for utility function $g$ with goal value $Q$, 
is an adaptive strategy for evaluating the absolute-majority function.   

We now apply the approximation bound for Adaptive Dual Greedy given above in (\ref{ADGbound})
to get an upper bound on the approximation factor ${\mathcal A}$ achieved by running ADG on the utility function $g$ that we constructed
for the absolute-majority function. 
More particularly, we will show that running ADG on utility function $g$ is an $O(d)$-approximation algorithm for that problem. 
We do not know whether this $O(d)$ approximation bound, for running ADG on $g$, is tight.

Let $R(x,S)$ denote the ratio in the expression for ${\mathcal A}$, so
$$R(x,S) := \frac{\sum_{i \in C(x)-S} [g(x,S \cup \{i\}) - g(x,S)]}{Q-g(x,S)}.$$
Consider the numerator in this fraction.
It is a summation over all $i \in C(x)-S$.
We will first prove an upper bound on the expression inside the summation,
$g(x,S\cup \{i\})-g(x,S)$.

Fix $x$ in $[d]^n$ and $S$ in $\mathbb{P}(x)$.  
In what follows, we will drop $x$ as a parameter in
the utility functions, as it will be equal to this fixed $x$.
For example, we will write $g(S)$ instead of $g(x,S)$. 

Let $j'$ be the value of $x_i$.
Then for $j \notin \{0,j'\}$, $g_j(S) = g_j(S \cup \{i\})$.
Intuitively, when we add $i$ to $S$, the distance to the goal $Q_j$ of each $g_j$ where $j \not\in \{0,j'\}$ does not change.
However, the distance to the goal $Q_{j'}$ of $g_{j'}$ decreases by 1.
The distance to the goal $Q_0$ of $g_0$ decreases by some amount that is between 1 and $d-1$.

For $j \in [d]$, 
let $Z_j := \prod_{\ell \notin \{0,j\}}(Q_{\ell}-g_{\ell}(S))$. 
Since $g_j(S) = g_j(S \cup \{i\})$ for $j \notin \{0,j'\}$, we have that the distance from $\hat{g}(S)$ to goal value $\hat{Q}$, and the distance of $\hat{g}(S \cup \{i\})$ to goal value $\hat{Q}$, are respectively:

\begin{equation}
\label{eq1}
    \hat{Q}-\hat{g}(S) = Z_{j'}(Q_{j'}-g_{j'}(S))
\end{equation}
and
\begin{equation}
\label{eq2}
    \hat{Q}-\hat{g}(S \cup \{i\}) = Z_{j'}(Q_{j'}-g_{j'}(S \cup \{i\})).
\end{equation}

Therefore,
\begin{equation}
\label{inbrackets}
\begin{split}
&g(S \cup \{i\}) - g(S)\\
&= (Q_0-g_0(S))(\hat{Q}-\hat{g}(S))-(Q_0-g_0(S \cup \{i\}))(\hat{Q}-\hat{g}(S \cup \{i\})\\
&= Z_{j'}[(Q_0-g_0(S))(Q_{j'}-g_{j'}(S))-(Q_0-g_0(S\cup \{i\}))(Q_{j'}-g_{j'}(S\cup\{i\}))]\\
&\leq Z_{j'}[(Q_0-g_0(S))(g_{j'}(S\cup\{i\})-g_{j'}(S)) + (Q_{j'}-g_{j'}(S))(g_0(S\cup\{i\})-g_0(S))],
\end{split}
\end{equation}
where the first equality is by the definition of $g$, the second is from (\ref{eq1}) and (\ref{eq2}),
and
we use the following simple property of non-negative numbers to justify the last inequality.  
Let $s,t,u,v \geq 0$.  Then 
$(s+t)(u+v) - tv \leq (s+t)u + (u+v)s$.
Take 
$s = g_0(S \cup \{i\}) - g_0(S)$ and
$t = Q_0-g_0(S \cup \{i\})$,
so $s+t=Q_0 - g_0(S)$.  Similarly, take
$u = g_{j'}(S \cup \{i\}) - g_{j'}(S)$ and
$v = Q_{j'}-g_{j'}(S \cup \{i\})$, so
$u+v=Q_{j'} - g_{j'}(S)$.
Substituting these values into the inequality $(s+t)(u+v) - tv \leq (s+t)u + (u+v)s$ yields the inequality above.

Interpreting the expression in brackets in the last line of (\ref{inbrackets}), as relating to the effect of adding $i$ to $S$, we see that it is equal to:

(Distance of $g_0$ from its goal before adding $i$)(Increase in $g_{j'}$ from adding $i$) + (Distance of $g_{j'}$ from its goal before adding $i$) (Increase in $g_0$ from adding $i$).

We now use the fact that the increase in $g_{j'}$ from adding $i$ is 1,
and the increase in $g_0$ from adding $i$ is at most $d-1$.
We thus have that 
$$g(S \cup \{i\}) - g(S) \leq Z_{j'}[(Q_0-g_0(S)) + (d-1)(Q_{j'}-g_{j'}(S))].$$

Recalling that we were using $j'$ to denote the value of $x_i$,
we therefore have the following upper bound on the numerator of $R(x,S)$:
\begin{equation}
\label{twoterms}
\sum_{i \in C(x)-S} [g(S \cup \{i\}) - g(S)]  \leq \Big[\sum_{i \in C(x)-S} Z_{x_i}(Q_0-g_0(S))\Big] + \Big[\sum_{i \in C(x)-S}(d-1)\prod_{j \in [d]} (Q_j-g_j(S))\Big].
\end{equation}

We will now separately bound the two terms on the right hand side of (\ref{twoterms}).
Recall that $C(x)$ is the sequence of tests performed by ADG on $x$, and that the last test performed is the one where the goal
value $Q$ of $g$ is reached.
It follows that for each $j \in [d]$, $C(x)-S$ can contain at most $Q_j-g_j(S)$ indices $i$ such that $x_i=j$, because after the tests
on $S$, an additional $Q_j-g_j(S)$ votes for candidate $j$ would cause
candidate $j$ to have enough votes for an 
 majority, meaning
that the goal value of $g_j$ would be reached.

It follows that for candidate $j$,
$$\sum_{i \in C(x)-S: x_i=j} Z_j(Q_0-g_0(S)) \leq Z_j(Q_0-g_0(S))(Q_j-g_j(S)) = Q-g(S).$$

This holds by symmetry for all candidates $j$, so 
$$\sum_{j \in [d]}\sum_{i \in C(S)-S:x_i=j} Z_j(Q_0-g_0(S)) \leq d(Q-g(S)).$$

The left-hand side of the above inequality is equal to
$\sum_{i \in C(S)-S} Z_{x_i}(Q_0-g_0(S))$, so we have the following upper bound on the first term on the right-hand side of (\ref{twoterms}):

$$\sum_{i \in C(S)-S} Z_{x_i}(Q_0-g_0(S)) \leq d(Q-g(S)).$$

We now bound the second term.  Each test performed after $S$ increases the value of $g_0$ by at least one, so the total number of tests in $C(x) - S$ is at most
$Q_0-g_0(S)$.  Therefore,
\begin{align*}
\sum_{i \in C(x)-S} (d-1)\prod_{j \in [d]} (Q_j-g_j(S)) &\leq (d-1)(Q_0-g_0(S))\prod_{j \in [d]} (Q_j-g_j(S)) \\
&= (d-1)(Q-g(S)).
\end{align*}

We have thus proved the following upper bound on the numerator of $R(x,S)$:
$$\sum_{i \in C(x)-S} g(S \cup \{i\}) - g(S) \leq d(Q-g(S)) + (d-1)(Q-g(S)) = (2d-1)(Q-g(S)).$$

Since the denominator of $R(x,S)$ is $Q-g(S)$, we have shown that 
${\mathcal A} \leq 2d-1$, meaning that we have proved our claimed approximation factor of $O(d)$.

\subsection{Use of ADG in Algorithm~\ref{alg:RMF-S}}

Algorithm~\ref{alg:RMF-S} uses Adaptive Dual Greedy after the votes obtained so far have narrowed down the number of potential relative-majority winners to two remaining candidates.
We discuss here how to use ADG to determine whether one of these candidates, say $\alpha$, is the relative-majority winner.

As stated in the description of the algorithm,
we use $b'$ to denote the partial assignment immediately after determining the two viable candidates $\alpha$ and $\beta$, and we need to determine whether or not
    $$N_\alpha(b')+\Big(\sum_{k\in\{1,\dots,d\}\setminus\{\beta\}}N_k(b')\Big)+\sum_{i:b_i'=*}y_i \geq n+1.$$
      Here $y_i=2$ if we find $x_i=\alpha$; $y_i=1$ if $x_i\in\{1,\dots,d\}\setminus\{\alpha,\beta\}$; and $y_i=0$ if $x_i=\beta$.

This is equivalent to finding the optimal adaptive strategy for evaluating a ternary linear threshold function with unit coefficients, that is, determining whether an inequality of the form 
$$\sum_{i \in [n]} x_i \geq \theta$$
holds, where each $x_i \in \{0,1,2\}$ and $\theta$ is a positive integer.  
Again, we address the evaluation problem in a stochastic environment where for each $x_i$ we are given $P[x_i=0]$, $P[x_i=1]$, and $P[x_i=2]$, the $x_i$ are independent, and it costs $c_i$ to obtain the value of $x_i$.

In order to apply ADG to this evaluation problem, we construct a utility function $g:([d] \cup \{*\})^n \rightarrow \mathbb{Z}^{\geq 0}
$ which is a minor variant of the utility function used by Deshpande et al.~\cite{deshpande2016approximation} in their algorithm for evaluating Boolean linear threshold functions. Function $g$ is built from two other utility functions, $g_1$ and $g_0$. The function $g_1$ is such that $g_1(b) = \min\{\theta,\sum_{i \in [n]:b_i \neq *} b_i\}$, while $g_0$ is such that $g_0(b) = \min\{2n-\theta+1,\sum_{i \in [n]: b_i \neq *} (2-b_i)\}$.  We define $Q_1=\theta$ and $Q_0 = 2n-\theta+1$.  We then combine $g_1$ and $g_0$ using the standard OR construction to obtain $g:\{0,1,2\}^n \rightarrow \mathbb{Z}^{\geq 0}$ such that for $b \in \{0,1,2,*\}^n$,
$$g(b) = Q_0Q_1 - (Q_0 - g_0(b))(Q_1-g_1(b)).$$

We define the goal value $Q$ for $g$ to be $Q_0Q_1$.   
To evaluate the ternary linear threshold function given above, we apply Adaptive Dual Greedy to achieve goal value $Q=Q_0Q_1$ for the constructed $g$.  
The approximation bound ${\mathcal A}$ given in (\ref{ADGbound}) for Adaptive Dual Greedy holds.  It remains to show that this bound is at most 3 for the constructed $g$.

As in the previous section,  for fixed $x \in \{0,1,2\}^n$, and $S \subseteq [n]$, we use $g(S)$ to represent the value of $g(b)$ where $b$ is such that $b_i = x_i$ for $i \in S$, and $b_i=*$ otherwise
(and similarly for $g_0$ and $g_1$.)
Again, we use $R(x,S)$ to denote the ratio in the expression for ${\mathcal A}$.

Now fix $x \in \{0,1,2\}^n$ and let $S$ be the set of tests in a fixed prefix of $C(x)$.
Let $R(x,S)$ be the ratio in the bound $A$ in (\ref{ADGbound}).  
Let $S'$ be the set of tests in $C(x)-S$.
Because ADG stops performing tests as soon as $g(b) = Q$, prior to performing the last test of $C(x)$ we have both $g_0(b) < Q_0$ and $g_1(b) < Q_1$.  
Without loss of generality, assume that 
the last test causes $g_1(b) = Q_1$.
In that case, the value of the last test must be 1 or 2.

If it is 2, then the final value of $b$ must
satisfy $g_1(b) = Q_1$ and $g_0(b) < Q_0$.
Further, by the definition of $g_1$,
and the fact that the $x_i$ are in $\{0,1,2\}$,
$\sum_{i:b_i\neq *} b_i \leq \theta+1$.
Also, $\sum_{i:i \in S'} (g_1(S \cup \{i\}) - g_1(S)) \leq  Q_1 - g_1(S) + 1$.
Since $g_0(b) < Q_0$, by the definition of $g_0$, $\sum_{i:b_i \neq *} b_i < 2n-\theta+1$
and $\sum_{i \in S'} (g_0(x,S \cup \{i\}) - g_0(x,S)) < Q_0 - g_0(S)$.

The denominator of $R(x,S)$ is equal to $(Q_1-g_1(S))(Q_0 - g_0(S))$.
The numerator is equal to 
$\sum_{i \in S'} (g(x,S \cup \{i\})-g(x,S))$.
For fixed $i \in S'$, a similar analysis to that used to show (\ref{inbrackets}) above yields 
\begin{align*}
&\sum_{i \in S'} [g(S \cup \{i\}) - g(S)] \\
= &\sum_{i \in S'} [(Q_0-g_0(S))(Q_1-g_1(S))-(Q_1-g_1(S\cup \{i\}))(Q_0-g_0(S\cup\{i\}))]\\
\leq &\sum_{i \in S'} (Q_0-g_0(S))(g_1(S\cup\{i\})-g_1(S)) + \sum_{i \in S'}(Q_1-g_1(S))(g_0(S\cup\{i\})-g_0(S))].
\end{align*}

From the above facts about $g_1$ and $g_0$, 
with our assumption that the last test had the value 2 and caused $g_1(b) = Q_1$,
we can upper bound the first term on the right-hand side of the above inequality:
$$\sum_{i \in S'} (Q_0-g_0(S))(g_1(S\cup\{i\})-g_1(S)) \leq (Q_0-g_0(S)) (Q_1 - g_1(S)+1).$$
We can also upper bound the second term:
$$\sum_{i \in S'}(Q_1-g_1(S))(g_0(S\cup\{i\})-g_0(S)) < (Q_1-g_1(S))(Q_0 - g_0(S)).$$
Since the denominator of $R(x,S)$ is $(Q_1-g_1(S))(Q_0 - g_0(S))$, it follows that
$$R(x,S) \leq 2 + 1 = 3.$$

The analysis in the case that the last test has the value 1 is similar.  In this case, 
we have that $\sum_{i \in S'} (g_0(x,S \cup \{i\}) - g_0(x,S)) \leq Q_0 - g_0(S)$
and $\sum_{i \in S'} (g_1(x,S \cup \{i\}) - g_0(x,S)) \leq Q_1 - g_1(S)$, implying that $R(x,S) \leq 2$.

\section{A $10$-approximation algorithm for evaluating the absolute-majority function with $2$ rounds of adaptivity.}\label{appendix:10-approx}

\begin{theorem}
    There exists a $10$-approximation algorithm for evaluating the absolute-majority function with $2$ rounds of adaptivity.
\end{theorem}

\begin{proof}
Here we propose an algorithm that replaces Phase 2 of Algorithm~\ref{alg:abs-maj-1} by testing the remaining votes following Phase 1 non-adaptively, in order of a permutation, until the value of $f$ if certain.  The permutation is generated using the modified round-robin approach introduced by Allen et al.~\cite{allen2017evaluation}. The modified round-robin algorithm creates a single permutation by combining the following four sequence permutations: increasing orders of $\frac{c_i}{p_{i,\alpha}},\frac{c_i}{1-p_{i,\alpha}}$, and $\frac{c_i}{p_{i,\beta}},\frac{c_i}{1-p_{i,\beta}}$, respectively. Let $\ALG$ denote this algorithm, and let $\OPT$ denote the optimal algorithm for evaluating the absolute-majority function.  The pseudocode for the modified round robin is given in Algorithm~\ref{alg:modified-round-robin}.

\begin{algorithm}[!t]
\caption{Modified round-robin (Allen et al.)}\label{alg:modified-round-robin}
\textbf{Input:} sequences $L_1,\dots,L_k$\;
Define $D[1,\dots,k]$\;
$D[i]\gets 0$ for $i\in\{1,\dots,k\}$\;
$\pi\gets []$\;
\While{$|\pi|<k\cdot n$}
{
    \For{$i\in\{1,\dots,k\}$}
    {
        $\tau_i\gets$ the index of the next variable in $L_i$ if exists; 0 otherwise\;
    }
    $i^*\gets$ any $j$ in $\argmin_{j:\tau_j \neq 0}(D[j]+c_{\tau_j})$\;
    Append $x_{\tau_{i^*}}$ to $\pi$\;
    $D[i^*]\gets D[i^*]+c_{\tau_{i^*}}$\;
}\smallskip
For each $x_i$ that appears more than once in $\pi$, keep the first occurrence in $\pi$ but delete all later occurrences.\;
\Return $\pi$
\end{algorithm}

Let $M$ denote the set of assignments where $\ALG$ needs to perform at least $1$ test in Phase 2. As $\alpha,\beta$ is known in Phase 2, we can further separate $M$ into $M_\alpha,M_\beta,M_0$, which represents the sets of assignments in $M$ that satisfy $f(a)=\alpha,f(a)=\beta,f(a)=0$ for all $a$ in the corresponding set, respectively. 
Let $C_1(\ALG,a)$ and $C_2(\ALG,a)$ denote the cost spent in Phase 1 and Phase 2 of our algorithm on assignment $a$, respectively. 

If we only consider the assignments in $M_\alpha$, what remains to do is to verify that $f=\alpha$. Then, an optimal strategy would test the remaining variables in increasing order of $c_i/p_{i,\alpha}$, and its cost is upper-bounded by the expected cost of $\OPT$ on $M_\alpha$. Because of the way the modified round-robin approach described in Algorithm~\ref{alg:modified-round-robin} is performed, for each assignment $a$, what we spend in Phase 2 of our algorithm is at most 4 times the cost of an optimal algorithm for verifying $f=\alpha$. (By an optimal algorithm for verifying $f=\alpha$, we mean its expected cost on random $x$, given that $f(x)=\alpha$, is as small as possible.) Hence, we have
$$\sum_{a\in M_\alpha}C_2(\ALG,a)p(a)\leq 4\sum_{a\in M_\alpha}C(\OPT,a)p(a).$$
And for assignments in $M_\beta$ and $M_0$, any strategy that evaluates the value of $f$ will first show that $f\neq \alpha$. Let $C_2^{f\neq \alpha}(\ALG,a)$ denote the cost of using $\ALG$ to verify $f\neq a$ on assignment $a$ in Phase 2. Recall that an optimal algorithm for verifying $f\neq \alpha$ tests the remaining variables in increasing order of $c_i/(1-p_{i,\alpha})$, and also its cost is again upper-bounded by the expected cost of $\OPT$ on $M_\beta\cup M_0$, we have
$$\sum_{a\in M_\beta\cup M_0}C_2^{f\neq \alpha}(\ALG,a)p(a)\leq 4\sum_{a\in M_\beta\cup M_0}C(\OPT,a)p(a).$$

For the same reason, we also have the following two inequalities
\begin{align*}
    \sum_{a\in M_\beta}C_2(\ALG,a)p(a)&\leq 4\sum_{a\in M_\beta}C(\OPT,a)p(a),\\
    \sum_{a\in M_\alpha\cup M_0}C_2^{f\neq \beta}(\ALG,a)p(a)&\leq 4\sum_{a\in M_\alpha\cup M_0}C(\OPT,a)p(a).
\end{align*}

Last, we consider the assignments in $M_0$. Since $C_2(\ALG,a)=\max\{C_2^{f\neq \alpha}(\ALG,a),C_2^{f\neq \beta}(\ALG,a)\}$ for $a\in M_0$, we have
\begin{align*}
    \sum_{a\in M_0}C_2(\ALG,a)p(a)=& \sum_{a\in M_0}(\max\{C_2^{f\neq \alpha}(\ALG,a),C_2^{f\neq \beta}(\ALG,a)\})p(a)\\
    \leq&\sum_{a\in M_0}(C_2^{f\neq \alpha}(\ALG,a)+C_2^{f\neq \beta}(\ALG,a))p(a).
\end{align*}

Combining what we have above, the expected cost of $\ALG$ spent in Phase 2 is
\begin{align*}
    \sum_{a\in [d]^n}C_2(\ALG,a)p(a)=&\sum_{a\in M}C_2(\ALG,a)p(a)+\sum_{a\in [d]^n\setminus M}C_2(\ALG,a)p(a)\\
    =&\sum_{a\in M_\alpha\cup M_\beta\cup M_0}C_2(\ALG,a)p(a)+\sum_{a\in [d]^n\setminus M}0\cdot p(a)\\
    \leq\;&8\sum_{a\in M}C(\OPT,a)p(a)\leq8\sum_{a\in [d]^n}C(\OPT,a)p(a).
\end{align*}
Since $C_1(\ALG,a)\leq 2C(\OPT,a)$ for all $a\in[d]^n$ as shown in the proof of Theorem~\ref{thm:abs-maj-1}, we have
\begin{align*}
    \sum_{a\in [d]^n}C(\ALG,a)p(a)=&\sum_{a\in [d]^n}C_1(\ALG,a)p(a)+\sum_{a\in [d]^n}C_2(\ALG,a)p(a)\\
    \leq\;&2\sum_{a\in [d]^n}C(\OPT,a)p(a)+8\sum_{a\in [d]^n}C(\OPT,a)p(a)=10\sum_{a\in [d]^n}C(\OPT,a)p(a).
\end{align*}
This completes the proof.
\end{proof}

\end{document}